\documentclass[11pt]{article}
\usepackage{paper}

\makeatletter
\newcommand{\TODO}[1]{\textcolor{red}{[TODO\@ifnotempty{#1}{: #1}]}}

\makeatother

\title{\textbf{Efficient algorithms for certifying lower bounds on the discrepancy of random matrices}}
\date{}
\author[1]{Prayaag Venkat\footnote{Email: \textit{pvenkat@g.harvard.edu}. Part of this work was done while visiting the Simons Institute for the Theory of Computing. Part of this work was done at Harvard, supported by an NSF Graduate Fellowship under grant DGE1745303 and Boaz Barak's Simons Investigator Fellowship, NSF grant DMS-2134157, DARPA grant W911NF2010021, and DOE grant DE-SC0022199, support from Oracle Labs and past support by the NSF, as well as the Packard and Sloan foundations and the BSF.}}

\begin{document}
\maketitle
\begin{abstract}
    We initiate the study of the algorithmic problem of certifying lower bounds on the discrepancy of random matrices: given an input matrix $A \in \R^{m \times n}$, output a value that is a lower bound on $\mathsf{disc}(A) = \min_{x \in \{\pm 1\}^n} \norm{Ax}_\infty$ for every $A$, but is close to the typical value of $\mathsf{disc}(A)$ with high probability over the choice of a random $A$. This problem is important because of its connections to conjecturally-hard average-case problems such as negatively-spiked PCA~\cite{bandeira2019computational}, the number-balancing problem~\cite{gamarnik2021algorithmic} and refuting random constraint satisfaction problems~\cite{raghavendra2017strongly}. We give the first polynomial-time algorithms with non-trivial guarantees for two main settings. First, when the entries of $A$ are \iid standard Gaussians, it is known that $\mathsf{disc} (A) = \Theta (\sqrt{n}2^{-n/m})$ with high probability~\cite{chandrasekaran2014integer,aubin2019storage,turner2020balancing} and that super-constant levels of the Sum-of-Squares SDP hierarchy fail to certify anything better than $\mathsf{disc}(A) \geq 0$ when $m < n - o(n)$~\cite{ghosh2020sum}. In contrast, our algorithm certifies that $\mathsf{disc}(A) \geq \exp(- O(n^2/m))$ with high probability. As an application, this formally refutes a conjecture of Bandeira, Kunisky, and Wein~\cite{bandeira2019computational} on the computational hardness of the detection problem in the negatively-spiked Wishart model. Second, we consider the \textit{integer partitioning problem}: given $n$ uniformly random $b$-bit integers $a_1, \ldots, a_n$, certify the non-existence of a \emph{perfect partition}, i.e. certify that $\mathsf{disc} (A) \geq 1$ for $A = (a_1, \ldots, a_n)$. Under the scaling $b = \alpha n$, it is known that the probability of the existence of a perfect partition undergoes a phase transition from 1 to 0 at $\alpha = 1$~\cite{borgs2001phase}; our algorithm certifies the non-existence of perfect partitions for some $\alpha = O(n)$. We also give efficient \textit{non-deterministic} algorithms with significantly improved guarantees, raising the possibility that the landscape of these certification problems closely resembles that of e.g. the problem of refuting random 3SAT formulas in the unsatisfiable regime. Our algorithms involve a reduction to the Shortest Vector Problem and employ the Lenstra-Lenstra-Lovász algorithm.
\end{abstract}

\section{Introduction}
The key object of study in this paper is the \emph{discrepancy} of a given matrix $A \in \R^{m \times n}$, defined as
\begin{equation}
\label{eqn:disc-defn}
\disc (A) = \min_{x \in \{\pm 1\}^n} \norm{Ax}_{\infty}.
\end{equation}
The problem of giving worst-case bounds on the discrepancy of matrices $A$ satisfying various assumptions has received intense study (see e.g. the books~\cite{matousek1999geometric,chazelle2000discrepancy,chen2014panorama} and references therein) and is connected to many fundamental problems in theoretical computer science, combinatorics, statistics and beyond. While much of past work has focused on proving such bounds non-constructively, recent research (see the survey~\cite{bansal2014algorithmic}) considers the algorithmic \emph{search problem}: given as input a matrix $A$, is there a polynomial-time algorithm that produces a signing $x \in \{\pm 1\}^n$ so that $\norm{Ax}_\infty$ is close to $\disc(A)$? 

Charikar, Newman and Nikolov~\cite{charikar2011tight} showed it is $\NP$-hard to distinguish between matrices $A \in \{0,1\}^{m \times n}$ with discrepancy zero and those with discrepancy $\Omega(\sqrt{n})$, when $m = O(n)$. Given this result, it is natural to study the discrepancy in an average-case setting in which $A$ is taken to be a random matrix. The study of this average-case setting is also motivated by the task of covariate balancing in randomized controlled trials~\cite{harshaw2019balancing,krieger2019nearly,turner2020balancing}. A sequence of works~\cite{karmarkar1982differencing,costello2009balancing,boettcher2008analysis,chandrasekaran2014integer,aubin2019storage,turner2020balancing} studying this problem has led to the following state-of-the-art non-algorithmic result~\cite{chandrasekaran2014integer,aubin2019storage,turner2020balancing}: if the entries of $A$ are \iid standard Gaussian random variables, then with high probability it holds that $\disc (A) = \Theta (\sqrt{n}2^{-n/m})$. Furthermore, Turner, Meka and Rigollet give a polynomial-time algorithm (which is a generalization of the classic Karmarkar-Karp algorithm~\cite{karmarkar1982differencing}) that finds a signing achieving discrepancy $\exp(-\Omega(\log^2 (n)/m))$ with high probability, provided that $m = O(\sqrt{\log(n)})$. This begs the question of whether or not this problem exhibits a statistical-to-computational gap: does there exist a polynomial-time algorithm that can compute with high probability a signing that achieves a discrepancy value at most $O(\sqrt{n}2^{-n/m})$ on Gaussian input $A$?

Recently, Gamarnik and K{\i}z{\i}lda{\u{g}}~\cite{gamarnik2021algorithmic} proved that in the $m=1$ setting, the set of signings achieving low discrepancy value for random $A$ satisfies the \emph{Overlap Gap Property} (OGP), which is thought to be an indicator of algorithmic hardness (see the survey~\cite{gamarnik2021overlap}). While they formally show that the class of ``stable'' algorithms fails to produce signings with discrepancy value smaller than $\exp(-\omega(n/\log^{1/5}(n)))$, they establish that OGP holds up to discrepancy value $\exp(-\omega(\sqrt{n \log (n)}))$. Using statistical physics-inspired techniques, several works~\cite{borgs2009proofi,borgs2009proofii,aubin2019storage,perkins2021frozen,abbe2022proof,gamarnik2022algorithms} have given evidence of the presence of statistical-to-computational gaps in average-case discrepancy problems.

Inspired by a rich body of research on the problem of certifying the unsatisfiability of random constraint satisfactions problems (CSPs) \cite{raghavendra2017strongly,kothari2017sum,bandeira2019computational,bandeira2021spectral}, we initiate the study of the algorithmic problem of efficiently \emph{certifying} lower bounds on the discrepancy of random matrices. More specifically, we ask: is there an efficient algorithm $\ALG$ which outputs a value $\ALG(A)$ on input $A$ such that for every $A$ it holds that $\ALG(A) \leq \disc(A)$, but for random $A$, $\ALG(A)$ is close to the true high-probability value $\Theta (\sqrt{n}2^{-n/m})$? While prior works have only focused on the search problem, we believe the certification problem is well-motivated for the following reasons. 

First, a natural approach to understanding the complexity of \emph{finding} low-discrepancy signings is to study the problem of \emph{distinguishing} a random matrix with \iid\,$\calN(0,1)$ entries that has discrepancy $\Theta(\sqrt{n} 2^{-n/m})$ with high probability from a random matrix  with a``planted'' signing that attains significantly smaller discrepancy. Bandeira, Kunisky, and Wein~\cite{bandeira2019computational} showed that a problem of this type, called the detection problem in the negatively-spiked Wishart model, is hard for the class of low-degree polynomial algorithms in some regime of parameters and conjectured that the same should be true for all polynomial-time algorithms. Observe that if an algorithm can solve the harder problem of certifying $\disc(A) \geq \delta$ with high probability for a Gaussian matrix $A$ and some $\delta > 0$, then it can distinguish such an $A$ from any family of matrices with discrepancy smaller than $\delta$.

Second, the certification problem has been thoroughly studied in the context of random CSPs \cite{raghavendra2017strongly,kothari2017sum,bandeira2019computational,bandeira2021spectral} and has connections to cryptography~\cite{applebaum2010public}, learning theory~\cite{daniely2016complexity}, and proof complexity~\cite{fleming2019semialgebraic}. This body of work has amassed strong evidence of the optimality of semidefinite programming (SDP)-based algorithms for a wide class of average-case problems exhibiting statistical-to-computational gaps. Given the relative scarcity of algorithms for solving average-case discrepancy problems, we hope that further study of the certification problem will inspire the development of novel algorithmic techniques and candidate optimal algorithms.

Finally, there is a long history of works in computer science and discrete mathematics designing efficent algorithms to complement non-constructive proofs of combinatorial results. The known proof that the discrepancy of a $m \times n$ Gaussian matrix is at least $\Omega(\sqrt{n} 2^{-m/n})$ with high probability makes use of the first-moment method. So, the naive algorithm to certify this fact simply enumerates the discrepancy values of all $2^n$ possible signings. For this reason, the problem of certifying average-case discrepancy lower bounds is non-trivial and thematically aligned with a large body of research that aims to characterize when non-constructive proofs can be made algorithmic.

Inspired by the success of convex relaxation techniques for certification problems in the context of random CSPs, it is natural to ask whether these techniques are applicable to certifying lower bounds on the discrepancy of random matrices. Interestingly, it is known that the $n^{\Omega(1)}$-degree Sum-of-Squares (SoS) SDP relaxation of~\ref{eqn:disc-defn} has value equal to $0$ with high probability when $m < n - o(n)$~\cite{ghosh2020sum}. Given the success of SoS in the context of random CSPs, this negative result begs the question of whether there exists \emph{any} polynomial-time algorithm certifying a value better than zero. In this paper, we give the first efficient certification algorithm with non-trivial guarantees for certifiying average-case discrepancy lower bounds.

\begin{theorem}
\label{thm:disc-poly}
There is an efficient deterministic algorithm that on input $A \in \R^{m \times n}$ with $m \leq n$ and \iid standard Gaussian entries certifies that $\disc(A) > \exp (- O(n^2 / m))$ \whp.
\end{theorem}
Theorem~\ref{thm:disc-poly} stands in sharp contrast to the state of affairs for certifying unsatisfiability of random CSPs, for which the SoS SDP hierarchy is believed to be the optimal algorithm. Furthermore, it immediately refutes a conjecture of Bandeira, Kunisky, and Wein~\cite{bandeira2019computational} on the computational hardness of the detection problem in the negatively-spiked Wishart model. In this problem, the goal is to distinguish which of the two following distributions a given matrix $A \in \R^{m \times n}$ was sampled from:

\begin{enumerate}
    \item (\textbf{Null}) The rows of $A$ are \iid samples from $\calN(0,I_n)$.
    \item (\textbf{Planted}) The rows of $A$ are \iid samples from $\calN(0,I_n - \frac{\beta}{n} vv^T)$, where $\beta < 1$ is the signal-to-noise ratio and $v$ is drawn uniformly at random from $\{\pm 1\}^n$.
\end{enumerate}

By taking orthogonal complements, one can verify that this problem is equivalent to detecting whether a random subspace contains a planted Boolean vector. Bandeira et al. showed that low-degree polynomial algorithms fail to distinguish these distributions when $\beta^2 < n/m$ and conjectured that this extends to all polynomial-time algorithms (Conjecture 3.1 of \cite{bandeira2019computational}). It is straightforward to verify that if $A$ is sampled from the planted distribution, then $\disc(A) \leq \polylog (m) \sqrt{n (1-\beta)}$ with high probability. Hence, for some $\beta \geq 1- \exp(-O(n^2/m))$, the algorithm from Theorem~\ref{thm:disc-poly} can distinguish the null and planted distributions, since the discrepancy under the planted distribution is strictly smaller than the discrepancy lower bound certified by the algorithm under the null distribution. However, we emphasize that our algorithm will only succeed for $\beta$ exponentially close to 1, so it is possible that a refined version of the conjecture of Bandeira et al. does still hold. Zadik, Song, Wein, and Bruna~\cite{zadik2022lattice} can formally solve the same problem when $\beta = 1$ and mention that their algorithm also likely works when $\beta$ is exponentially close to 1.

We also mention an interesting phenomenon regarding efficient \emph{non-deterministic} certification. Here, an efficient non-deterministic certification algorithm is one that produces a polynomial \emph{size} witness that exists with high probability for a random matrix $A$, does not exist for any low-discrepancy $A$ and can be verified (but not necessarily computed) in polynomial time. The existence of such an algorithm is an average-case analogue of being in the complexity class $\coNP$. A fascinating result of Feige, Kim and Ofek~\cite{feige2006witnesses} shows the existence of polynomial \emph{size} certificates of the unsatisfiability of random 3SAT formulas on $n$ variables and $m = O(n^{1.4})$ clauses, whereas it is strongly believed that polynomial \emph{time} algorithms for certifying unsatisfiability can only succeed when $m = \Omega(n^{1.5})$. We leave open the possibility of a similar phenomenon occuring in the context of average-case discrepancy.

\begin{theorem}
\label{thm:disc-nondet}
There is an efficient non-deterministic algorithm that on input $A \in \R^{m \times n}$ with $m \leq n$ and \iid standard Gaussian entries certifies that $\disc(A) \geq \exp (- O(n \log (n) / m))$ \whp.
\end{theorem}

While we do not prove any algorithmic hardness results in this paper, the previous two theorems raise the possibility of a regime of parameters in which there are succinct certificates of discrepancy lower bounds, yet there are no efficient algorithms to \emph{find} these certificates. In the language of Feige, Kim and Ofek, this means that for the average-case complexity of this discrepancy problem, ``$\coNP \subseteq \mathsf{P}$'' for $\delta > \exp (-O(n^2/m))$ and ``$\coNP \subseteq \NP$ ''for $\delta > \exp(-O(n \log(n)/m))$.

We now turn our attention to the \emph{integer partitioning problem}, a generalization of one of the six original $\NP$-complete problems of Garey and Johnson~\cite{garey1979computers}, for which a similar story takes place. Given $n$ uniformly random $b$-bit integers $a_1, \ldots, a_n$, the integer partitioning problem asks to find a perfect partition, i.e. a subset $S \subseteq [n]$ such that 
\[
\left| \sum_{i \in S} a_i - \sum_{i \notin S} a_i \right| \leq 1.
\]
This is nothing but an average-case discrepancy problem in disguise; rescaling $\Bar{A} = a/2^b$ (to be thought of as a vector of $b$-bit truncations of $\Unif([0,1])$ random variables), a perfect partition exists if and only if $\disc(\Bar{A}) \leq 2^{-b}$. The integer partitioning problem has been studied thoroughly in both the computer science~\cite{karmarkar1982differencing,gent1996phase,korf1998complete} and statistical physics~\cite{mertens1998phase,mertens2000random,borgs2001phase,bauke2004number,bauke2004universality,borgs2009proofi,borgs2009proofii} communities and was among the first average-case combinatorial optimization problems for which phase transition behavior was fully characterized. Under the scaling $b = \alpha n$, Borgs, Chayes, and Pittel~\cite{borgs2001phase} showed that the probability of existence of a perfect partition undergoes a phase transition: when $\alpha < 1$ a perfect partition exists \whp and when $\alpha > 1$, no perfect partition exists \whp. Motivated by the previous discussion, we ask: what is the smallest value of $\alpha > 1$ for which there is an efficient algorithm that certifies the absence of perfect partitions? To the best of our knowledge, this question has not been studied before. We give analogues of Theorems~\ref{thm:disc-poly} and \ref{thm:disc-nondet}, the first non-trivial certification guarantees for integer partitioning.

\begin{theorem}
\label{thm:ip-poly}
There is an efficient deterministic algorithm that on input $a_1, \ldots, a_n \in \{0, \ldots, 2^{b} - 1\}$ drawn \iid uniformly at random certifies that no perfect partition of $a_1, \ldots a_n$ exists \whp when $b = \alpha n$ for some $\alpha = O (n)$.
\end{theorem}

\begin{theorem}
\label{thm:ip-nondet}
There is an efficient non-deterministic algorithm that on input $a_1, \ldots, a_n \in \{0, \ldots, 2^{b} - 1\}$ drawn \iid uniformly at random for $b = \alpha n$ certifies that no perfect partition of $a_1, \ldots a_n$ exists \whp for some $\alpha = O (\log (n))$.
\end{theorem}

\subsection{Techniques}
At a high level, our algorithms reduce the problem of certifying lower bounds on discrepancy to the problem of deciding whether a certain lattice contains a short vector. To approximately solve this instance of the shortest vector problem (SVP) in polynomial time, we invoke the Lenstra–Lenstra–Lovász (LLL) algorithm. In fact, our result allows one to translate, in a black-box way, the approximation guarantee of any given SVP oracle $\calO$ to the discrepancy lower bound cerified by our algorithm instantiated with $\calO$. 

We also remark that lattice basis reduction techniques have recently been used to solve \emph{search} versions of various average-case problems exhibiting conjectural statistical-to-computational gaps~\cite{gamarnik2021inference,song2021cryptographic,diakonikolas2022non,zadik2022lattice}. These works are not directly comparable to the present paper for two reasons. First, they study \emph{search} problems, whereas we study \emph{certification}; in general, there is no formal connection between the two and in some cases, their complexities can be quite different (see~\cite{bandeira2019computational} for a notable example). Second, while they do not directly apply to the negatively-spiked Wishart model, they can solve a greater variety of problems, such as non-Gaussian component analysis~\cite{diakonikolas2022non}, clustering Gaussian mixtures~\cite{zadik2022lattice} and various other noiseless inference problems~\cite{gamarnik2021inference,song2021cryptographic}. Our results complement this line of work by demonstrating the utility of lattice-based techniques for solving \emph{certification} problems as well. For both search and certification problems, lattice basis reduction techniques break computational barriers that apply to other classes of algorithms like low-degree polynomials and SoS.

\subsection{Future work}
In this work, we gave the first non-trivial algorithms for two fundamental average-case certification problems. We bring to bear a novel algorithmic technique for the certification problem that outperforms standard convex relaxation techniques. While we focused on Gaussian and integer input settings for simplicity, we believe it is straightforward to extend our results to a broader class of distributions satisfying mild concentration and anti-concentration properties. 

Our results leave open the possibility of a statistical-to-computational gap for certifying discrepancy lower bounds, mirroring the scenario for random CSPs. An important direction for future research is to either design algorithms which improve on those in this paper or provide rigorous evidence for hardness of average-case certification of discrepancy lower bounds. This is a particularly challenging task because it is currently unclear whether such gaps can be predicted by analyzing a restricted class of algorithms. We have no reason to believe that the algorithms in this paper are optimal; any improvement on the value certified in, say, Theorem~\ref{thm:disc-poly} would be very interesting. Again taking inspiration from the study of certifying unsatisfiability of random CSPs~\cite{raghavendra2017strongly}, we ask: can one design a \emph{sub-exponential} time algorithm that certifies a better value than the value given in Theorem~\ref{thm:disc-poly}?

We conclude by mentioning another related open problem regarding the discrepancy of \emph{Bernoulli} matrices. It is known that  the probability that an $m \times n$ matrix $A$ with \iid $\mathsf{Bernoulli(1/2)}$ entries will have discrepancy at most 1 undergoes a phase transition from 0 to 1 at $n = \Theta (m \log m)$~\cite{potukuchi2018discrepancy} and that it will have discrepancy $\Theta(\sqrt{n})$ with high probability when $n = \Theta(m)$. Altschuler and Niles-Weed~\cite{altschuler2022discrepancy} conjecture that no efficient algorithm can even find a constant discrepancy signing in the regime $n \geq Cm \log m$. We pose the following certification problem: what is the largest $n$ for which there is an efficient algorithm that certifies the discrepancy of a random binary matrix is strictly bigger than 1 with high probability? Unfortunately, the algorithms in this paper do not apply to this Bernoulli model.

\section{Preliminaries}
\subsection{Computational model}
We now specify the details of the computational model in which our algorithms operate. Let $A \in \R^{m \times n}$ be a matrix whose entries are \iid according to some distribution $\calD$ on $\R$ and $b \in \N$ be a truncation parameter. The algorithm receives as input the matrix $\Bar{A} \in \R^{m \times n}$ whose $(i,j)$ entry is $A_{ij}$ truncated to $b$ bits of precision, for every $i \in [m], j \in [n]$. We say that an algorithm $\ALG$ certifies a discrepancy lower bound of $\delta = \delta(m,n,b)$ on $\Bar{A}$ if:
\begin{itemize}
    \item For every input $\Bar{A}$, $\ALG$ outputs a value $\ALG (\Bar{A})$ such that $\ALG (\Bar{A}) \leq \disc (\Bar{A})$.
    \item For random input $\Bar{A}$ generated as described above, $\ALG (\Bar{A}) \geq \delta$ with high probability.
\end{itemize}
In the Gaussian setting (i.e. $\calD = \calN (0,1)$), our algorithm works in a model in which it can query the $b$ most significant bits in the binary representation of any entry of $A$ at $O(b)$ computational cost, for any $b$. 

For the integer partitioning problem (i.e. $\calD = \Unif ([0,1])$), the algorithm is simply given as input the $b$-bit representations of the numbers $a_1, \ldots, a_n$ for some value of $b$ that it cannot choose. Furthermore, we say that an algorithm certifies the non-existence of a perfect partition if for every instance $a_1, \ldots, a_n$, the algorithm never reports that no perfect partition exists if one does exist.

\subsection{SVP}
Given a collection of linearly independent vectors $\calB = \{b_1, \ldots, b_k\} \subset \R^d$, the \emph{lattice} $\calL = \calL(\calB)$ generated by basis vectors in $\calB$ is defined as 
\[
\calL = \left \{\sum_{i=1}^k x_i b_i: x \in \Z^k \right \}.
\]
For any lattice $\calL$, we can define the length of its shortezt non-zero vector as
\[
\lambda_1(\calL) = \min_{x \in \calL \setminus \{0\}} \norm{x}_2.
\]
The $\gapsvp_\alpha$ problem is to distinguish, given an input lattice $\calL$ (described by its basis) and parameter $\alpha \geq 1$, whether $\lambda_1(\calL) \leq 1$ or $\lambda_1(\calL) > \alpha$, under the promise that $\calL$ satisfies exactly one of these two conditions. The main algorithm in this work requires an oracle for the $\gapsvp_\alpha$ problem; we now state the guarantees of two algorithms for SVP.

\begin{theorem}[\cite{lenstra1982factoring}]
\label{theorem:lll}
There is a deterministic algorithm that given input collection $\calB = \{b_1, \ldots, b_k \} \subset \Q^d$ of linearly independent vectors with bit complexity $b$ solves the $\gapsvp_\alpha$ problem on instance $\calL(\calB)$ for $\alpha = 2^{k/2}$ in time $\poly(k,d,b)$.
\end{theorem}

\begin{theorem}[\cite{aharonov2005lattice}]
\label{theorem:svp-conp}
There is some constant $c > 0$ such that for any instance $\calL(\calB)$, described by a collection $\calB = \{b_1, \ldots, b_k \} \subset \Q^d$ of linearly independent vectors with bit complexity $b$, of $\gapsvp_\alpha$ with $\alpha = c \sqrt{k}$, there is a non-deterministic algorithm that produces a $\poly(k,d,b)$-time verifiable certificate of either $\lambda_1(\calL (\calB)) \leq 1$ or $\lambda_1(\calL (\calB)) > \alpha$.
\end{theorem}

\subsection{John ellipsoid}
Let $M \in \R^{m \times n}$ and define the centrally-symmetric polytope $P_M = \{x \in \R^n : \norm{Mx}_\infty \leq 1\}$ and the ball of radius $r$ as $B(0,r) = \{x \in \R^n: \norm{x}_2 \leq r\}$. John's Theorem guarantees the existence of an invertible linear transformation $T$ such that:
\[
B(0,1) \subseteq T(P_M) \subseteq B(0,\sqrt{n}).
\]
We will require an efficient algorithm for approximately computing such a $T$.

\begin{theorem}[Theorem 1.1 of \cite{cohen2019near}]
\label{thm:john-ellipsoid}
There is an efficient algorithm that given input $M \in \R^{m \times n}$ with $\poly(m,n)$-bit entries outputs an invertible linear transformation $T \in \R^{n \times n}$ satisfying:
\begin{equation}
\label{eqn:john-condition}
B(0,1) \subseteq T(P_M) \subseteq B(0,2\sqrt{n}).    
\end{equation}
\end{theorem}

\section{Certifying discrepancy lower bounds}
The key subroutine in the algorithms behind Theorems~\ref{thm:disc-poly}-\ref{thm:ip-nondet} is Algorithm~\ref{alg:cert}. In this section, we state and analyze Algorithm~\ref{alg:cert}, after which the proofs of the main theorems will follow easily. The following lemma verifies that Algorithm~\ref{alg:cert} correctly certifies lower bounds on discrepancy.
 
\begin{figure}[htbp]
    \centering\begin{mdframed}[style=algo] 
\begin{enumerate}
    \item \textbf{Input:} Matrix $\Bar{A} \in \R^{m \times n}$ with $b$-bit entries and rows $\Bar{A}_1, \ldots, \Bar{A}_m$, parameters $\delta > 0, \alpha \geq 1$ and $\gapsvp_\alpha$ oracle $\calO$.
    
    \item Define $S_\delta = \{x \in \R^n: \lvert \ip{\Bar{A}_i}{x} \rvert \leq \delta, i=1,\ldots,m\} \cap [-1,1]^n$.
    
    \item Compute an invertible linear transformation $T \in \R^{n \times n}$ such that $B(0,1) \subseteq T(S_\delta) \subseteq B(0, 2\sqrt{n})$ (using the algorithm from Theorem~\ref{thm:john-ellipsoid}). 
    
    \item Define the lattices $\calL = T(\Z^n)$ and $\calL' = \frac{1}{2\sqrt{n}}\calL$.
    
    \item Query $\calO$ on input lattice $\calL'$. Output $\delta$ if $\lambda_1(\calL') > \alpha$. Otherwise, output $0$.
    
    \item \textbf{Output:} Value $\ALG(\Bar{A})$ that satisfies $\ALG(\Bar{A}) \leq \disc(\Bar{A})$.
\end{enumerate}
\end{mdframed}
\caption{Certification algorithm}
\label{alg:cert}
\end{figure}

\begin{lemma}[Correctness]
\label{lemma:correctness}
On any input $\Bar{A} \in \R^{m \times n}$ with $b$-bit entries and any $\delta > 0$, Algorithm~\ref{alg:cert} satisfies $\ALG (\Bar{A}) \leq \disc (\Bar{A})$. 
\end{lemma}
\begin{proof}
To prove the claim, it suffices to the consider the case that $\disc(\Bar{A}) < \delta$. In this case, there exists $x \in \{\pm 1\}^n$ such that $\norm{\Bar{A} x}_\infty < \delta$. In particular, it holds that $x \neq 0$ and $x \in S_\delta \cap \Z^n$. Next, note that $Tx \neq 0$ (by invertibility of $T$) and $Tx \in T(S_\delta) \cap \calL \subseteq B(0, 2\sqrt{n}) \cap \calL$. Together, these imply that $\lambda_1(\calL') \leq 1$.
By correctness of the $\gapsvp_\alpha$ oracle, Algorithm~\ref{alg:cert} will return $0$, so we may conclude $\ALG (\Bar{A}) \leq \disc (\Bar{A})$.
\end{proof}

Lemma~\ref{lemma:cert-value} below characterizes the high-probability value certified by Algorithm~\ref{alg:cert} in terms of various parameters of the input distribution.

\begin{lemma}
\label{lemma:cert-value}
Assume that $m \leq n$ and let $A \in \R^{m \times n}$ with entries \iid according to a continuous distribution $\calD$ with density bounded by 1 and tail function $Q_\calD (t) = \Prob_{z \sim \calD}(|z| > t)$. Next, suppose there are $\delta > 0$, $b \in \N$, $M \in [0, 2^b-1]$ so that the following conditions are satisfied:
\begin{enumerate}
    \item $\delta \leq \exp (-4 n \log (2 \alpha \sqrt{n})/m)$
    \item $Q_{\calD} (M) = o(n^{-1} (2 \alpha \sqrt{n} + 1)^{-n/m})$
    \item $b \geq 2 \log_2 (M \alpha n) + 2n \log_2 (2\alpha \sqrt{n} + 1)/m$
\end{enumerate}
Then on input $\Bar{A} \in \R^{m \times n}$ which is the entry-wise $b$-bit truncation of $A$, Algorithm~\ref{alg:cert} satisfies $\ALG (\Bar{A}) \geq \delta$ \whp.
\end{lemma}

In order to prove Lemma~\ref{lemma:cert-value}, we use the following result concerning the anti-concentration of the rows of the input matrix.

\begin{lemma}
\label{lemma:anti-conc}
Let $A \in \R^{n}$ have coordinates \iid according to a continuous distribution $\calD$ with density bounded by 1, tail function $Q_\calD$,  $\Bar{A}$ be its $b$-bit truncation, and $y \in \R^n$ be any vector satisfying $1 \leq \norm{y}_\infty \leq \beta$. Next, let $\theta \geq 0$ and $M \geq 0$. Then, we have:
\[
\Prob (\lvert \ip{\Bar{A}}{y} \rvert \leq \theta) \leq 2 \theta + \frac{2 M n \beta}{2^b} + n Q_{\calD} (M).
\]
\end{lemma}

\begin{proof}[Proof of Lemma~\ref{lemma:cert-value}]
To prove the claim, it suffices to show that $\lambda_1(\calL) > 2 \alpha \sqrt{n}$ \whp on random input $\Bar{A}$. By definition, this means $B(0,2 \alpha \sqrt{n}) \cap \calL = \{0\}$, which is in turn implied by $T^{-1}(B(0, 2 \alpha \sqrt{n})) \cap \Z^n = \{0\}$ (because $T$ is invertible). By Condition~\ref{eqn:john-condition}, it holds that $T^{-1}(B(0, 2 \alpha \sqrt{n})) = 2 \alpha \sqrt{n} T^{-1} (B(0,1)) \subseteq 2 \alpha \sqrt{n} S_\delta$. The proof will be complete by showing that $2\alpha \sqrt{n} S_\delta \cap \Z^n = \{0\}$ with high probability. Define the set $\calU = \Z^n \cap [-2\alpha \sqrt{n}, 2\alpha \sqrt{n}]^n \setminus \{0\}$. Then we conclude with:
\begin{align*}
\Prob (2 \alpha \sqrt{n} S_\delta \cap \Z^n \neq \{0\}) &\leq | \calU | \cdot \max_{y \in \calU} \Prob (\left| \ip{\Bar{A}_1}{y} \right| \leq 2 \alpha \sqrt{n} \delta)^m \\
&\leq |\calU | \cdot (4 \alpha \sqrt{n} \delta + \frac{4 M \alpha n^{1.5}}{2^b} + n Q_\calD(M))^m \\
&\leq (2\alpha \sqrt{n} + 1)^n \cdot (4 \alpha \sqrt{n} \delta + \frac{4 M \alpha n^{1.5}}{2^b} + n Q_\calD(M))^m = o(1)
\end{align*}
where the second inequality follows from Lemma~\ref{lemma:anti-conc} and the final equality follows from the assumptions on $\delta, M, b$.
\end{proof}

\begin{proof}[Proof of Lemma~\ref{lemma:anti-conc}]
Defining $E$ to be the event that $\norm{A}_\infty \leq M$, we have that:
\begin{align*}
\Prob (\lvert \ip{\Bar{A}}{y}\rvert \leq \theta) &= \Prob(E) \Prob (\lvert \ip{\Bar{A}}{y}\rvert \leq \theta|E) + \Prob (E^c) \Prob (\lvert \ip{\Bar{A}}{y}\rvert \leq \theta |E^c) \\
&\leq \Prob (\lvert \ip{\Bar{A}}{y} \rvert \leq \theta | E) + \Prob (E^c).
\end{align*}
By a union bound and definition of the tail function, the second term is upper bounded by $n Q_\calD (M)$. To control the first term, note that on the events $E$ and $\lvert\ip{\Bar{A}}{y} \rvert \leq \theta$, 
\begin{align*}
\lvert\ip{A}{y}\rvert &\leq \lvert \ip{\Bar{A}}{y} \rvert + \lvert \ip{A - \Bar{A}}{y} \rvert \\
&\leq \theta + \norm{A - \Bar{A}}_\infty \norm{y}_1 \\
&\leq \theta + \norm{A - \Bar{A}}_\infty \cdot n \norm{y}_\infty \\
&\leq \theta + \frac{M n \beta}{2^b}.
\end{align*}
Setting $\theta' = \theta + \frac{M n \beta}{2^b}$ and assuming $y_1 \geq 1$, without loss of generality, we can control the first term as follows:
\begin{align*}
\Prob (\lvert \ip{\Bar{A}}{y} \rvert \leq \theta | E) &\leq \Prob (\lvert \ip{A}{y} \rvert \leq \theta'| E) 
\\
&\leq \Prob \left( - \sum_{i=2}^n A_i \frac{y_i}{y_1} - \frac{\theta'}{y_1} \leq A_1 \leq - \sum_{i=2}^n A_i \frac{y_i}{y_1} + \frac{\theta'}{y_1} \middle| E \right)\\
&\leq \frac{2 \theta'}{y_1} \leq 2 \theta'.
\end{align*}
\end{proof}

\subsection{Proofs of main results}
Equipped with the above technical lemmas, we now prove Theorems~\ref{thm:disc-poly} and \ref{thm:ip-poly}. The proofs of Theorems~\ref{thm:disc-nondet} and \ref{thm:ip-nondet} follow in the same way, but by using the algorithm in Theorem~\ref{theorem:svp-conp} as a $\gapsvp_\alpha$ oracle instead of the LLL algorithm. 

\begin{proof}[Proof of Theorem~\ref{thm:disc-poly}]
The certification procedure is as follows:
\begin{enumerate}
    \item First, set $u = O(\log (mn))$ and certify that $|A_{ij}| \leq u$ for all $i \in [m], j \in [n]$ by inspecting the first $b = 8n^3$ bits of each entry $A_{ij}$; record this $b$-bit truncation in $\Bar{A}_{ij}$. If for some $(i,j)$ it holds that $|\Bar{A}_{ij}| > u$, then output 0. Otherwise, proceed to the next step.
    \item Run Algorithm~\ref{alg:cert} with parameters $\alpha = 2^{n/2}, \delta = \exp(-6n^2 /m), b = 8n^3$ and the LLL algorithm as a $\gapsvp_\alpha$ oracle on input $\Bar{A}$. Output $\ALG (\Bar{A}) - \frac{u n}{2^b}$.
\end{enumerate}
To prove correctness of the procedure, we show that for any $A$, the value it outputs for instance $A$ is a lower bound on $\disc(A)$. If there is $i \in [m], j \in [n]$ such that $|A_{ij}| > u$, then the procedure outputs $0$ (which trivially lower bounds $\disc (A)$) in the first step. If $|A_{ij}| \leq u$ for all $i \in [m], j \in [n]$, then $\disc (A) \geq \disc (\Bar{A}) - \frac{u n}{2^b}$. By Lemma~\ref{lemma:correctness}, we also have $\disc(\Bar{A}) \geq \ALG (\Bar{A})$. Hence, the value output by the procedure is always a lower bound on $\disc(A)$.

Next, note that the runtime of the procedure is dominated by the approximate John ellipsoid computation and the call to the $\gapsvp_\alpha$ oracle in Algorithm~\ref{alg:cert}. By Theorems~\ref{theorem:lll} and \ref{thm:john-ellipsoid}, each of these steps can be implemented in deterministic $\poly(m,n,b)$ time.

We now analyze the high-probability value certified by this procedure $A_{ij} \sim \calN(0,1)$. Invoking Lemma~\ref{lemma:cert-value} with $\calD = \calN (0,1)$, $M = \Theta (n \log(n)/ \sqrt{m})$, $Q_\calD (M) \leq 2 \exp (-M^2/2)$, which satisfy the hypotheses, we conclude that the procedure certifies $\disc(A) \geq \exp (- O(n^2  / m))$ \whp.
\end{proof}

\begin{proof}[Proof of Theorem~\ref{thm:ip-poly}]
First, observe that if $a_1, \ldots, a_n \in \{0, 2^{b} - 1\}$, then certifying that no perfect partition of $a_1, \ldots, a_n$ exists is equivalent to certifying that $\disc(\Bar{A}) > 2^{-b}$, where $\Bar{A} = (a_1/2^b, \ldots, a_n/ 2^b)$. Note also that $\Bar{A}$ has the same distribution as that of the entrywise $b$-bit truncation of a random vector $A \in \R^n$ with entries \iid according to $\Unif ([0,1])$. To certify the non-existence of a perfect partition, run Algorithm~\ref{alg:cert} on input $\Bar{A}$ with parameters $\alpha = 2^{n/2}$, $\delta = 2^{-b}$ and the LLL algorithm as a $\gapsvp_\alpha$ oracle. If $\ALG (\Bar{A}) > 0$, then report that no perfect partition exists.

The correctness of this procedure follows immediately from Lemma~\ref{lemma:correctness}: if on any input $a_1, \ldots, a_n$ it holds that $\ALG (\Bar{A}) > 0$, then $\ALG (\Bar{A}) = \delta = 2^{-b}$ and $\disc(\Bar{A}) > 2^{-b}$. Hence, the procedure will never report the non-existence of a perfect partition if one exists.

As in the proof of Theorem~\ref{thm:disc-poly}, the runtime of the procedure is dominated by the approximate John ellipsoid computation and the call to the $\gapsvp_\alpha$ oracle, each of which can be implemented in deterministic $\poly(n,b)$ time.

Next, we show that if $b$ is sufficiently large, then the procedure reports the non-existence of a perfect partition with high probability. Invoking Lemma~\ref{lemma:cert-value} with $\calD = \Unif ([0,1])$, $M = 1$, $Q_\calD (M) = 0$, which satisfy the hypotheses provided $b \geq C n^2$ for some sufficiently large absolute constant $C > 0$, we conclude that \whp, Algorithm~\ref{alg:cert} on input $\Bar{A}$ outputs $\delta = 2^{-b} > 0$.

\end{proof}

\section*{Acknowledgments}
The author would like to thank Boaz Barak, Alex Wein and Ilias Zadik for helpful discussions.

\newpage
\bibliographystyle{alpha}
\bibliography{bib}
\end{document}